\documentclass{llncs}

\usepackage{epsfig}
\usepackage{amssymb}
\usepackage{amsmath}
\usepackage{graphics}
\usepackage{pgf,tikz}
\usetikzlibrary{arrows,decorations.pathmorphing,backgrounds,fit,positioning,patterns} 
\usepackage{enumitem}
\usepackage{color}
\usepackage{xcolor}
\usepackage{multirow}
\usepackage{subfig}
\usepackage[normalem]{ulem}
\usepackage{bold-extra}

\usepackage{amsthm}
\usepackage{breqn}

\def\dnsitem{\vspace{2pt}\item}

\newcommand\wght{\omega}
\newcommand\calA{\mathcal{A}}

\newcommand\calP{\mathcal{P}}

\begin{document}

\title{Tight Approximation Bounds for the Seminar Assignment Problem 
\thanks{This work is supported in part by grants from NSF CNS 1302563, by Navy N00014-16-1-2151, by NSF CNS 1035736, by NSF CNS 1219064. Any opinions, findings, and conclusions or recommendations expressed here are those of the authors and do not necessarily reflect the views of sponsors.}
}

\author{Amotz Bar-Noy  \and George Rabanca }
\institute{Department of Computer Science, Graduate Center, CUNY, New York, USA\\
}

\maketitle

\abstract{The seminar assignment problem is a variant of the generalized assignment problem in which items have unit size and the amount of space allowed in each bin is restricted to an arbitrary set of values.  The problem has been shown to be NP-complete and to not admit a PTAS.  However, the only constant factor approximation algorithm known to date is randomized and it is not guaranteed to always produce a feasible solution.  

In this paper we show that  a natural greedy algorithm outputs a solution with value within a factor of $(1 - e^{-1})$ of the optimal, and that unless $NP\subseteq DTIME(n^{\log\log n})$, this is the best approximation guarantee achievable by any polynomial time algorithm.
}



\section{Introduction}
In the \textsc{Seminar Assignment} problem (SAP) introduced in \cite{Krumke13} one is given a set of seminars (or bins) $B$, a set of students (or items) $I$, and for each seminar $b$ a set of integers $K_b$ specifying the allowable number of students that can be assigned to the seminar. Unless otherwise specified, we assume that $0 \in K_b$ for any $b \in B$.   For each student $i$ and seminar $b\in B$ let $p(i, b) \in \mathbb{R}$ represent the profit generated from assigning student $i$ to seminar $b$.  A {\em seminar assignment} is a function $\calA: J \rightarrow B$ where $J\subseteq I$ and we say that the assignment is feasible if  $|\calA^{-1}(b)| \in K_b$ for all $b\in B$, where $\calA^{-1}$ is the pre-image of $\calA$.  The goal is to find a feasible seminar assignment $\calA$ that maximizes the total profit: 
$$
p(\calA) = \sum_{i\in J} p(i, \calA(i)).
$$

The problem has been introduced in \cite{Krumke13} in a slightly less general version.  In the original version, for each $b\in B$ the set $K_b$ equals to $\{0\} \cup \{l_b, ..., u_b\}$ for some lower and upper bounds $l_b, u_b \in \mathbb{N}$.  The more general setting considered in this paper can be useful for example when a seminar doesn't just require a minimum number of students and has a fixed capacity, but in addition requires students to work in pairs and therefore would allow only an even number of students to be registered.   In addition, this generalization also simplifies notation.

SAP is a variant of the classic \textsc{General Assignment} problem (GAP) in which one is given $m$  bins with capacity $B_1, ..., B_m$ and $n$ items.  Each item $i$ has size $s(i, b)$ in bin $b$ and yields profit $p(i, b)$.  The goal is to find a packing of the items into the bins that maximizes total profit, subject to the constraint that no bin is overfilled.  A GAP instance with a single bin is equivalent to the knapsack problem, and a GAP instance with unit profit can be interpreted as a decision version of the bin packing problem: can all items be packed in the $m$ bins?

SAP is also related to the \textsc{Maximum Coverage} problem (MC).  In the classic version of the MC problem one is given a collection of sets $\mathcal{S} = \{S_1, ..., S_m\}$ and a budget $B$.  The goal is to select a subcollection $\mathcal{S}' \subseteq \mathcal{S}$ with cardinality less than or equal to $B$ such that $|\cup_{S \in \mathcal{S}'} S|$ is maximized.

The algorithms with the best approximation ratio for both MC and GAP are greedy algorithms and the approximation bounds have been proved with similar techniques.  In this paper we show how to extend these analysis techniques to SAP.

\subsubsection{Related Work.}
In \cite{Krumke13} the authors show that SAP is NP-complete even when $K_b = \{0, 3\}$ for all $b\in B$ and $p(i, b) \in \{0, 1\}$ for any $i\in I$.  Moreover, they show that SAP does not admit a PTAS by providing a gap-preserving reduction from the 3-bounded 3-dimensional matching problem.  In \cite{Bender13} the authors investigate the approximability of the problem and provide a randomized algorithm which they claim outputs a solution that in expectation has value at least $1/3.93$ of the  optimal.  In \cite{Bender13_2} this result is revised and the authors show that for any $c \geq 2$, their randomized algorithm outputs a feasible solution with probability at least $1 - \min\{\frac{1}{c}, \frac{e^{c-1}}{c^c}\}$ and has an approximation ratio of $\frac{e-1}{(2c-1)\cdot e}$.

The GAP is well studied in the literature, with \cite{CATTRYSSE92} and \cite{Kundakcioglu2009} surveying the existing algorithms and heuristics for multiple variations of the problem.  In \cite{Shmoys93} the authors provide a $2$-approximation algorithm for the problem and in \cite{Cohen06} it is shown that any $\alpha$-approximation algorithm to the knapsack problem can be transformed into a $(1 + \alpha)$-approximation algorithm for GAP.  In \cite{Fleischer06} tight bounds for the GAP are given showing that no polynomial time algorithm can guarantee a solution within a factor better than $(1 - e^{-1})$, unless $P = NP$, and  providing an LP-based approximation which for any $\epsilon > 0$ outputs a solution with profit within a $(1 - e^{-1}-\epsilon)$ factor of the optimal solution value.

The GAP with minimum quantities, in which a bin cannot be used if it is not packed at least above a certain threshold, is introduced in \cite{Krumke13}.  Because items have arbitrary size, it is easy to see that when a single bin is given and the lower bound threshold equals the bin capacity, finding a feasible solution with profit greater than zero is equivalent to solving \textsc{Subset Sum}.  Therefore, in its most general case the problem cannot be approximated in polynomial time, unless $P = NP$.

In \cite{Nemhauser1981} and \cite{Conforti84} the authors study the problem of maximizing a non-decreasing submodular function $f$ satisfying $f(\emptyset) = 0$ under a cardinality constraint.  They show that a simple greedy algorithm achieves an approximation factor of $(1 - e^{-1})$ which is the best possible under standard assumptions.  Vohra and Hall note that the classic version of the maximum coverage problem belongs to this class of problems \cite{Vohra93}.  When each set $S_i$ in the MC problem is associated with a cost $c(S_i)$ the \textsc{Budgeted Maximum Coverage} problem asks to find a collection of sets $\mathcal{S}'$ covering the maximum number of elements under the (knapsack) constraint that $\sum_{S_i \in \mathcal{S}'} c(S_i) \leq B$ for some budget $B \in \mathbb{R}$.  In \cite{Khuller99} the authors show that the greedy algorithm combined with a partial enumeration of all solutions with small cardinality also achieves a $(1 - e^{-1})$ approximation guarantee, and provide matching lower bounds which hold even in the setting of the classic MC problem (when all sets have unit cost).  In \cite{Sviridenko04} Sviridenko generalizes the algorithm and proof technique to show that maximizing any monotone submodular function under a knapsack constraint can be approximated within $(1 - e^{-1})$ as well.

\subsubsection{Contributions.}
In Section \ref{sec:hardness}, by a reduction from the \textsc{Maximum Coverage} problem, we show that there exists no polynomial time algorithm that guarantees an approximation factor larger than $(1 - e^{-1})$, unless $NP\subseteq DTIME(n^{\log\log n})$.  
In Section \ref{sec:greedy} we present a  greedy algorithm that outputs a solution that has profit at least $\frac{1}{2}\cdot(1 - e^{-1})$ of the optimal solution.  The algorithm is based on the observation that when the required number of students in each seminar is fixed, the problem is solvable in polynomial time.  Finally, in Section \ref{sec:greedy_plus} we show how this algorithm can be improved to guarantee an approximation bound of $(1 - e^{-1})$. 

\section{Hardness of Approximation}\label{sec:hardness}
In this section we show that the problem is hard to approximate within a factor of $(1 - e^{-1} + \epsilon)$, $\forall \epsilon > 0$, even for the case when for each $b\in B$ the set $K_b$ equals $\{0, n\}$ for some integer $n$, and the profit for assigning any student to any seminar is either $0$ or $1$.  We prove this result by showing that such restricted instances of SAP are as hard to approximate as the \textsc{Maximum Coverage} problem defined below.

\begin{definition}
Given a collection of sets $\mathcal{S} = \{S_1, ..., S_m\}$ and an integer $k$, the \textsc{Maximum Coverage} (MC) problem is to find a collection of sets $\mathcal{S}' \subseteq \mathcal{S}$ such that $|\mathcal{S}'| \leq k$ and the union of the sets in $\mathcal{S}'$ is maximized.
\end{definition}

In \cite{Khuller99} it is shown that the MC problem is hard to approximate within a factor of $(1 - e^{-1} + \epsilon)$, unless $NP\subseteq DTIME(n^{\log\log n})$.  We use this result to prove the following:

\begin{theorem}
For any $\epsilon > 0$ the SAP is hard to approximate within a factor of $(1 - e^{-1} + \epsilon)$ unless $NP\subseteq DTIME(n^{\log\log n})$.
\end{theorem}
\begin{proof}
To prove the theorem we create a SAP instance for any given MC instance and show that from any solution of the SAP instance we can create a solution for the MC instance with at least equal value, and that the optimal solution of the SAP instance has value at least equal to the optimal solution of the MC instance.  Therefore, an $\alpha$-approximation algorithm for SAP can be transformed into an $\alpha$-approximation algorithm for MC.

Given a MC instance, let $U = \cup_{S \in \mathcal{S}} S$ and $n = |U|$.  For each set $S\in \mathcal{S}$ let $b_S$ be a seminar with the allowable number of students $K_b = \{0, n\}$, and for each element $e \in U$ let $i_e$ be a student in $I$.  The profit of a student $i_e$ assigned to a seminar $b_S$ is $1$ if the element $e$ belongs to the set $S$ and $0$ otherwise.  In addition, let $d_1, ..., d_{n*(k-1)}$ be dummy students that have profit $0$ for any seminar.

We first show that any feasible assignment $\calA$ corresponds to a valid solution to the given MC instance.  Since every seminar requires exactly $n$ students and there are exactly $k\cdot n$ students available, clearly at most $k$ seminars can be assigned students in any feasible assignment.  Let $\mathcal{S}' = \{S \in \mathcal{S}: \calA(b_S) > 0\}$.  It is easy to see that the number of elements in $\cup_{S \in S'} S$ is at least equal to the profit $p(\calA)$ since a student $i_e$ has profit $1$ for a seminar $b_S$ only if the set $S$ covers element $e$.

It remains to show that for any solution to the MC instance there exists a solution to the corresponding SAP instance with the same value.  Fix a collection of sets  $\mathcal{S}' \subseteq \mathcal{S}$ with $|\mathcal{S}'| \leq k$.  For every $e\in \cup_{S \in S'} S$ let $S_e$ be a set in $\mathcal{S}'$ that contains $e$ and let $\calA(i_e) = b_{S_e}$.  Then, assign additional dummy students to any seminar with at least one student to reach the required $n$ students per seminar.   Clearly, the profit of the assignment $\calA$ is equal to the number of elements covered by the collection $\mathcal{S}'$, which proves the theorem.
\end{proof}

\section{Seminars of Fixed Size}
In this section we show that when the allowable number of students that can be assigned to any seminar $b$ is a set $K = \{0, k_b\}$ for some integer $k_b$, SAP can be approximated within a factor of $(1- e^{-1})$ in polynomial time.  This introduces some of the techniques used in the general case in a simpler setting.

For an instance of the SAP, a {\em seminar selection} is a function $S:B \rightarrow \mathbb{N}$ with the property that $S(b) \in K_b$ for any $b\in B$.  We say that $S$ is feasible if $\sum_{b \in B} S(b) \leq |I|$.  In other words, a seminar selection is a function that maps each seminar to the number of students to be assigned to it.  A seminar selection $S$ {\em corresponds} to an assignment $\calA$ if for any seminar $b$ the number of students assigned by $\calA$ to $b$ is $S(b)$.  We slightly abuse notations and denote by $p(S)$ the maximum profit over all seminar assignments corresponding to the seminar selection $S$; we call $p(S)$ the profit of $S$.  In the remainder of this paper for a graph $G = (V, E)$ we denote the subgraph induced by the vertices of $X\subseteq V$ by $G[X]$.

\begin{definition}
Given a SAP instance let $V_b=\{v_{b,1}, ..., v_{b,k_b}\}$ for every $b\in B$ and let $V = \cup_{b\in B} V_b$.  The {\em bipartite representation} of the instance is the complete bipartite graph  $G=(V\cup I, E)$ with edge weights $\wght(v_b, i) = p(i, b)$ for every $v_b \in V_b$.  The {\em bipartite representation} of a seminar selection $S$ is the graph $G[V_S\cup I]$ where $V_S = \cup_{b \in B} V_{S, b}$ and $V_{S, b} =\{v_{b,1}, ..., v_{b,S(b)}\}$ for every $b\in B$.
\end{definition}

\begin{lemma}\label{lem:opt_selection}
For any SAP instance and any feasible seminar selection $S$, $p(S)$ is equal to the value of the maximum weight matching in the bipartite representation of $S$.
\end{lemma}
\begin{proof}
Let $G_S=(V_S \cup I, E)$ be the bipartite representation of $S$. First observe that any matching $M$ of $G_S$ that matches all the vertices of $V_S$ can be interpreted as an assignment $\calA_M$ of equal value by setting $\calA_M(i) = b$  whenever vertex $i\in I$ is matched by $M$ to a vertex in $V_{S,b}$.  Since $G_S$ is complete and has non-negative edge weights, there exists a maximum weight matching that matches all the vertices of $V_S$.

Similarly, any feasible assignment for the SAP instance can be interpreted as a matching $M_{\mathcal{A}}$ of equal value, which proves the lemma.
\end{proof}

\begin{definition} For a given finite set $A$, a set function $f:2^A \rightarrow \mathbb{R}$ is submodular if for any $X, Y \subseteq A$ it holds that: 
$$
f(X) + f(Y) \geq f(X \cup Y) + f(X \cap Y).  
$$
\end{definition}

Sviridenko shows that certain submodular functions can be maximized under knapsack constraints, which will be useful in proving Theorem \ref{thm:fixed_size}:

\begin{theorem}[\cite{Sviridenko04}]\label{thm:max_submod}
Given a finite set $A$, a submodular, non-decreasing, non-negative, polynomially computable function $f:2^A \rightarrow \mathbb{R}$, a budget $L \geq 0$, and costs $c_a \geq 0$, $\forall a \in A$, the following optimization problem is approximable within a factor of $(1-e^{-1})$ in polynomial time: 
$$
\max_{X\subseteq A}\left\{f(X): \sum_{x\in X} c_x \leq L\right\}
$$
\end{theorem}

We relate now the value of a maximum weight matching in a bipartite graph to the notion of submodularity.  
\begin{definition}
For an edge weighted bipartite graph $G = (A \cup B, E)$, the {\em partial maximum weight matching function} $f:2^A \rightarrow \mathbb{R}$ maps any set $S\subseteq A$ to the value of the maximum weight matching in $G[S\cup B]$.
\end{definition}

\begin{lemma}\label{lemma:submod}
Let $f$ be the partial maximum weight matching function for a bipartite graph $G = (A \cup B, E)$ with non negative edge weights.  Then $f$ is submodular.
\end{lemma}
\begin{proof}
Fix two sets $X, Y \subseteq A$ and let $M_\cap$ and $M_\cup$ be two matchings for the graphs $G[(X\cap Y) \cup B]$ and $G[(X \cup Y) \cup B]$ respectively.  To prove the lemma it is enough to show that it is possible to partition the edges in $M_\cap$ and $M_\cup$ into two disjoint matchings $M_X$ and $M_Y$ for the graphs $G[X\cup B]$ and $G[Y\cup B]$ respectively.

The edges of $M_\cap$ and $M_\cup$ form a collection of alternating paths and cycles.  Let $\mathcal{C}$ denote this collection and observe that no cycle of $\mathcal{C}$ contains vertices from $X\setminus Y$ or $Y\setminus X$.   This holds because $M_\cap$ does not match those vertices.

Let $\mathcal{P}_X$ be the set of paths in $\mathcal{C}$ with at least one vertex in $X \setminus Y$ and let $\mathcal{P}_Y$ be the set of paths in $\mathcal{C}$ with at least one vertex in $Y \setminus X$.  Two such paths are depicted in Fig. \ref{fig:submodular}.

{\em Claim 1.}
$\mathcal{P}_X \cap \mathcal{P}_Y = \emptyset$.  

{\em Proof of claim:}
Assume by contradiction that there exists a path $P \in \mathcal{P}_X \cap \mathcal{P}_Y$.  Let $x$ be a vertex in $X\setminus Y$ on path $P$ and similarly let $y$ be a vertex in $Y \setminus X$ on path $P$.  Observe that since neither $x$ nor $y$ belong to $X \cap Y$ they do not belong to the matching $M_\cap$ by definition, and therefore they are the endpoints of the path $P$.  Moreover, since both $x$ and $y$ are in $A$, the path $P$ has even length and since it is an alternating path, either the first or last edge belongs to $M_\cap$. Therefore $M_\cap$ matches either $x$ or $y$ contradicting its definition.  \qed

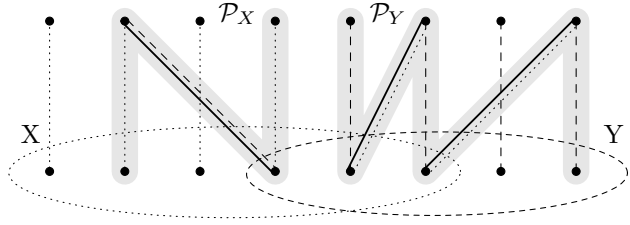
\begin{figure}
\captionsetup{width=.84\textwidth}
\center
\begin{tikzpicture}[line cap=round,line join=round,>=triangle 45,x=1.0cm,y=1.0cm]
\clip(0.0,0.0) rectangle (10.0,3.4);
\draw [color=black!10,line width=10pt] (2, 1)-- (2, 3)-- (4, 1)-- (4, 3);
\draw [color=black!10,line width=10pt] (5, 3)-- (5, 1)-- (6, 3)-- (6,1)-- (8, 3)-- (8, 1);

\draw [rotate around={0.19488285223860294:(3.4600000000000053,0.9900000000000002)},dotted] (3.4600000000000053,0.9900000000000002) ellipse (3.0016901159056184cm and 0.6053458118509876cm);
\draw [rotate around={0.23196545129820545:(6.150000000000015,0.9700000000000001)},dash pattern=on 2pt off 2pt] (6.150000000000015,0.9700000000000001) ellipse (2.531702775321735cm and 0.5554448150552177cm);
\draw [line width=.7pt] (3.95,1.0)-- (1.95,3.0);
\draw [dashed] (4.05,1.0)-- (2.05,3.0);
\draw [line width=.7pt] (4.95,1.0)-- (5.95,3.0);
\draw [dotted] (5.05,1.0)-- (6.05,3.0);
\draw [line width=.7pt] (5.95,1.0)-- (7.95,3.0);
\draw [dotted] (6.07,1.0)-- (8.07,3.0);

\draw [dotted] (1.0,3.0)-- (1.0,1.0);
\draw [dotted] (2.0,3.0)-- (2.0,1.0);
\draw [dotted] (3.0,3.0)-- (3.0,1.0);
\draw [dotted] (4.0,3.0)-- (4.0,1.0);
\draw [dashed] (5.0,3.0)-- (5.0,1.0);
\draw [dashed] (6.0,3.0)-- (6.0,1.0);
\draw [dashed] (7.0,3.0)-- (7.0,1.0);
\draw [dashed] (8.0,3.0)-- (8.0,1.0);
\begin{scriptsize}
\draw [fill=black] (1.0,1.0) circle (1.5pt);
\draw [fill=black] (1.0,3.0) circle (1.5pt);
\draw [fill=black] (2.0,1.0) circle (1.5pt);
\draw [fill=black] (2.0,3.0) circle (1.5pt);
\draw [fill=black] (3.0,1.0) circle (1.5pt);
\draw [fill=black] (3.0,3.0) circle (1.5pt);
\draw [fill=black] (4.0,1.0) circle (1.5pt);
\draw [fill=black] (4.0,3.0) circle (1.5pt);
\draw [fill=black] (5.0,1.0) circle (1.5pt);
\draw [fill=black] (5.0,3.0) circle (1.5pt);
\draw [fill=black] (6.0,1.0) circle (1.5pt);
\draw [fill=black] (6.0,3.0) circle (1.5pt);
\draw [fill=black] (7.0,1.0) circle (1.5pt);
\draw [fill=black] (7.0,3.0) circle (1.5pt);
\draw [fill=black] (8.0,1.0) circle (1.5pt);
\draw [fill=black] (8.0,3.0) circle (1.5pt);
\draw (.75, 1.5) node {\normalsize X};
\draw (8.5, 1.5) node {\normalsize Y};
\draw (5.5, 3.1) node {\normalsize $\mathcal{P}_Y$};
\draw (3.5, 3.1) node {\normalsize $\mathcal{P}_X$};
\end{scriptsize}
\end{tikzpicture}
\caption{$M_{X\cup Y}$ matches each vertex in $X\cup Y$ to the vertex directly above it.  $M_{X\cap Y}$ is depicted with contiguous segments, $M_{X}$ with dotted segments and $M_Y$ with dashed segments.  Two alternating paths of $\mathcal{P}$ are shown in light gray.}
\label{fig:submodular}
\end{figure}

For a set of paths $\mathcal{P}$ we let $E(\mathcal{P}) = \{e\in P: P\in \mathcal{P}\}$.
Moreover, let 
$$M_X = (E(\mathcal{P}_X) \cap M_\cup) \cup ( E(\mathcal{C} \setminus \mathcal{P}_X) \cap M_\cap)$$ 
and 
$$M_Y = (E(\mathcal{P}_X) \cap M_\cap) \cup ( E(\mathcal{C} \setminus \mathcal{P}_X) \cap M_\cup).$$
It is clear that $M_X \cup M_Y = M_\cap\cup M_\cup$  and $M_X \cap M_Y = M_\cap \cap M_\cup$.  To prove the theorem it remains to show that $M_X$ and $M_Y$ are valid matchings for $G[X\cup B]$ and $G[Y\cup B]$ respectively.
To see that $M_X$ is a valid matchings for $G[X\cup B]$ observe first that that no vertex of $Y \setminus X$ is matched by $M_X$ since $\mathcal{P}_X$ does not intersect $Y \setminus X$ by Claim 1, and $M_\cap$ does not intersect $Y \setminus X$ by definition.  Therefore, $M_X$ only uses vertices of $X \cup B$.  Second observe that every vertex $x\in X$ is matched by at most one edge of $M_X$ since otherwise $x$ belongs to either two edges of $M_\cup$ or two edges of $M_\cap$, contradicting the definition.  This proves that $M_X$ is a valid matching for $G[X\cup B]$;  showing that $M_Y$ is a valid matchings for $G[Y\cup B]$ is similar.
\end{proof}

\begin{theorem}\label{thm:fixed_size}
Any instance of SAP in which $|K_b|\leq 2$ for all $b \in B$ can be approximated in polynomial time to a factor of $(1-e^{-1})$.
\end{theorem}
\begin{proof}
Fix a SAP instance and for any $X\subseteq B$ let $S_X$ be the seminar selection which allocates $k_b$ students to any seminar in $S$ and $0$ students to any seminar in $B\setminus S$.  Moreover, let $G$ be the bipartite representation of the SAP instance and $f$ be the partial maximum weight matching function for graph $G$.  Denote by $G[V_X\cup I]$ the bipartite representation of $S_X$ and let $g(X)=f(V_X)$.  Since $f$ is submodular by Lemma \ref{lemma:submod}, it is easy to see that $g$ is submodular as well.  Assume by contradiction that there exist sets $X, Y \subseteq B$ such that the submodularity condition for $g$ doesn't hold: 
\begin{align}\label{eq:sap_submod}
g(X) + g(Y) < g(X\cup Y) + g(X \cap Y).
\end{align}
Therefore, by definition of $g$ we have
$$
f(V_X) + f(V_Y) < f(V_X\cup V_Y) + g(V_X \cap V_Y),
$$
contradicting the submodularity of $f$ proven in Lemma \ref{lemma:submod}.

Clearly $g$ is also monotone, non-negative and polynomially computable.  Let $c_b = k_b$, $\forall b \in B$, let $L = |I|$, and observe that $S_X$ is feasible if and only if $\sum_{x\in X} c_x \leq L$.  Moreover, by Lemma \ref{lem:opt_selection} and the definition of $g$, $g(X) = p(S_X)$ whenever the seminar selection $S_X$ is feasible and therefore the proof follows from Theorem \ref{thm:max_submod}.
\end{proof}

\section{A Constant Factor Greedy Algorithm}\label{sec:greedy}
The algorithm presented in this section sequentially increments the number of students allocated to each seminar in a greedy fashion.  It is similar in nature to the greedy algorithm of \cite{Khuller99} and \cite{Sviridenko04} but the details of the approximation guarantee proof are different.  
In the rest of this section we denote by $\calA_S$ an optimal assignment for the seminar selection $S$.  Remember that Lemma \ref{lem:opt_selection} shows that given feasible seminar selection $S$, an optimal seminar assignment $\calA_S$ can be found in polynomial time.

We say that a seminar selection $T$ is greater than a selection $S$ (denoted by $T \succ S$) if $T(b) \geq S(b)$, $\forall b\in B$, and there exists $b \in B$ s.t.  $T(b) > S(b)$.  The cost of a seminar selection $S$ is denoted by $c(S)$ and equals $\sum_{b\in B} S(b)$.  When $T \succ S$ we define the {\em marginal cost} of $T$ relative to $S$ as the difference between the cost of $T$ and the cost of $S$: 
$$
c_S(T) =  c(T) - c(S)
$$

Similarly, we define $p_S(T) = p(T) - p(S)$, the {\em marginal profit} of $T$ relative to $S$.  We say that $T$ is an {\em incrementing selection} for a seminar selection $S$ if $T \succ S$ and there exists a single seminar for which the selection $T$ allocates more students than selection $S$;  more precisely, the cardinality of the set $\{b\in B: T(b) > S(b)\}$ is $1$.  For a selection $S$ we denote the set of incrementing seminar selections that are feasible by $inc(S)$.  

We are now ready to present our algorithm: 

{\center
\fbox{
\begin{minipage}{280pt}
\noindent {\bf \textsc{Greedy}} 
\begin{enumerate}
\item $S_0 =$ initial seminar selection; 
\item $i = 0$;
\dnsitem While $inc(S_i) \neq \emptyset$: 
\begin{enumerate}
	\dnsitem $S_{i+1} \leftarrow \arg\max_{S' \in inc(S_i)} (p(S') - p(S_i)) / (c(S') - c(S_i))$;
	\dnsitem $i \leftarrow i + 1$
\end{enumerate}
\dnsitem $\calA_1 \leftarrow \calA_{S_i}$;
\dnsitem $\calA_2 \leftarrow $ maximum assignment to any single seminar $b$ for which $S_0(b) = 0$;
\dnsitem Return $\max{\calA_1, \calA_2}$;
\end{enumerate}
\end{minipage}}
}
\vspace{10pt}

In this section we analyze the algorithm starting from an empty initial seminar selection.  In the following section we show that by running the algorithm repeatedly with different initial seminar selections, the approximation guarantee can be improved.

Observe that the cardinality of $inc(S)$ is never greater than $|B| \cdot |I|$ and is therefore polynomial in the size of the input.  Thus, using the maximum weight matching reduction from the proof of Lemma \ref{lem:opt_selection}, step 3(a) of the algorithm can be performed efficiently.  

\begin{definition}
For a seminar selection $S$ and a tuple $(b, k_b)$ with $b \in B$ and $k_b \in \mathbb{N}$, let $S \oplus (b, k_b)$ denote the seminar selection $S'$ with $S'(b) = \max \{k_b, S(b)\}$ and $S'(b') = S(b')$ for any $b'\in B$, $b' \neq b$.
\end{definition}


\begin{lemma}\label{lem:difference}
For any feasible seminar selections $S$ and $T$, if for every seminar $b \in B$ the seminar selection $S \oplus (b, T(b))$ is feasible, then it holds that:
$$
\sum_{b \in B} \left[p(S \oplus (b, T(b))) - p(S) \right] \geq p(T) -  p(S).
$$
\end{lemma}
\begin{proof}
 

For a fixed SAP instance let $G$ be its bipartite representation and let $G[V_S\cup I]$ and $G[V_T\cup I]$ be the bipartite representations of $S$ and $T$ respectively.  Moreover, let $M_S$ and $M_T$ be two maximum weight matchings in $G[V_S\cup I]$ and $G[V_T\cup I]$ respectively.  Remember that according to Lemma \ref{lem:opt_selection} it holds that $p(S) = \wght(M_S)$ and $p(T) = \wght(M_T)$.  To prove the lemma we create matchings $\mathcal{M} = \{M_b\}_{b\in B}$ for the bipartite representations of assignments $p(S \oplus (b, T(b))$, such that each edge of $M_T$ is used in exactly one of the matchings in $\mathcal{M}$ and each edge of $M_S$ is used in exactly $|B| - 1$ of the matchings in $\mathcal{M}$.

Let $\mathcal{C}$ be the collection of isolated components formed by the union of the edges of $M_S$ and $M_T$.  Since both $M_S$ and $M_T$ are matchings in $G$, each element of $\mathcal{C}$ is a path or cycle in $G$.  For every $b \in B$ let $\mathcal{P}_b = \{P \in \mathcal{C}: V(P) \cap V_b \cap (V(M_T) \setminus V(M_S)) \neq \emptyset\}$, where $V(P)$ denotes the vertices of component $P$ (Fig. \ref{fig:path_partition}).

{\em Claim 2.}  For any $a \neq b \in B$, $\calP_a \cap \calP_b = \emptyset$. 

{\em Proof of claim:}
To prove the claim, assume that there exist $P \in \calP_a \cap \calP_b$ for some $a \neq b \in B$.
   Then by definition there exist $v_a \in V_a$ and $v_b \in V_b$ such that $v_a, v_b \in V(P)$ and $v_a, v_b \notin V(M_S)$ and therefore $v_a$ and $v_b$ are the endpoints of the alternating path $P$.  Since neither of the endpoints of the path belong to $M_S$, $P$ must have an odd number of edges.  However, because both endpoints of $P$ belong to the same partition of the bipartite graph $G$, the path $P$ must have an even number of edges, hence the claim holds by contradiction. \qed

\begin{figure}[]
\center
\begin{tabular}{cccc}
\subfloat[]{
\definecolor{uququq}{rgb}{0.6,0.6,0.6}
\begin{tikzpicture}[line cap=round,line join=round,>=triangle 45,x=1.0cm,y=1.0cm,scale=.65]
\clip(-0,-0.4) rectangle (4.8,8.0);

\draw[color=black!10,line width=10pt]  (1.0,7)-- (4.0,7);
\draw[color=black!10,line width=10pt]  (1.0,6)-- (4.0,6)-- (1,5)-- (4,5)-- (1,4)-- (4,4);
\draw[color=black!10,line width=10pt]  (4.0,3)-- (1.0,3)-- (4,2)-- (1,2)-- (4,1)-- (1,1)-- (4,0);

\draw[color=uququq,fill=uququq,fill opacity=0.1] (0.75,7.25) -- (0.75,4.75) -- (1.25,4.75) -- (1.25,7.25) -- cycle;
\draw[color=uququq,fill=uququq,fill opacity=0.1] (0.75,4.25) -- (0.75,2.75) -- (1.25,2.75) -- (1.25,4.25) -- cycle;
\draw[color=uququq,fill=uququq,fill opacity=0.1] (0.75,2.25) -- (0.75,0.75) -- (1.25,0.75) -- (1.25,2.25) -- cycle;

\draw (-0,7) node[anchor=north west] {$b_1$};
\draw (-0,4) node[anchor=north west] {$b_2$};
\draw (-0,2) node[anchor=north west] {$b_3$};
\draw (4,7) node[anchor=north west] {$P_1$};
\draw (4,4.8) node[anchor=north west] {$P_2$};
\draw (1.2,.5) node[anchor=north west] {$P_3$};

\draw [dotted] (1.0,7.05)-- (4.0,7.05);
\draw [dotted] (1.0,6.0)-- (4.0,6.0);
\draw [dotted] (1.0,5.0)-- (4.0,5.0);
\draw [dotted] (1.0,4.0)-- (4.0,4.0);
\draw [dotted] (1.0,3.0)-- (4.0,3.0);
\draw [dotted] (1.0,2.0)-- (4.0,2.0);
\draw [dotted] (1.0,1.0)-- (4.0,1.0);

\draw [] (1.0,6.95)-- (4.0,6.95);
\draw [] (1.0,5.0)-- (4.0,6.0);
\draw [] (1.0,4.0)-- (4.0,5.0);
\draw [] (1.0,3.0)-- (4.0,2.0);
\draw [] (1.0,2.0)-- (4.0,1.0);
\draw [] (1.0,1.0)-- (4.0,0.0);

\begin{scriptsize}
\draw [fill=black] (1.0,1.0) circle (1.5pt);
\draw [fill=black] (1.0,2.0) circle (1.5pt);
\draw [fill=black] (1.0,3.0) circle (1.5pt);
\draw [fill=black] (1.0,4.0) circle (1.5pt);
\draw [fill=black] (1.0,5.0) circle (1.5pt);
\draw [fill=black] (1.0,6.0) circle (1.5pt);
\draw [fill=black] (1.0,7.0) circle (1.5pt);
\draw [fill=black] (4.0,7.0) circle (1.5pt);
\draw [fill=black] (4.0,6.0) circle (1.5pt);
\draw [fill=black] (4.0,5.0) circle (1.5pt);
\draw [fill=black] (4.0,4.0) circle (1.5pt);
\draw [fill=black] (4.0,3.0) circle (1.5pt);
\draw [fill=black] (4.0,2.0) circle (1.5pt);
\draw [fill=black] (4.0,1.0) circle (1.5pt);
\draw [fill=black] (4.0,0.0) circle (1.5pt);
\end{scriptsize}
\end{tikzpicture}
} & \subfloat[] {
\definecolor{uququq}{rgb}{0.6,0.6,0.6}
\begin{tikzpicture}[line cap=round,line join=round,>=triangle 45,x=1.0cm,y=1.0cm,scale=.65]
\clip(-0,-0.4) rectangle (3.2,8.0);

\draw[color=uququq,fill=uququq,fill opacity=0.1] (0.75,7.25) -- (0.75,4.75) -- (1.25,4.75) -- (1.25,7.25) -- cycle;

\draw (-0,7) node[anchor=north west] {$b_1$};

\draw [dotted] (1.0,7.05)-- (3.0,7.05);
\draw [dotted] (1.0,6.0)-- (3.0,6.0);
\draw [dotted] (1.0,5.0)-- (3.0,5.0);
\draw [dotted] (1.0,4.0)-- (3.0,4.0);

\draw [] (1.0,3.0)-- (3.0,2.0);
\draw [] (1.0,2.0)-- (3.0,1.0);
\draw [] (1.0,1.0)-- (3.0,0.0);

\begin{scriptsize}
\draw [fill=black] (1.0,1.0) circle (1.5pt);
\draw [fill=black] (1.0,2.0) circle (1.5pt);
\draw [fill=black] (1.0,3.0) circle (1.5pt);
\draw [fill=black] (1.0,4.0) circle (1.5pt);
\draw [fill=black] (1.0,5.0) circle (1.5pt);
\draw [fill=black] (1.0,6.0) circle (1.5pt);
\draw [fill=black] (1.0,7.0) circle (1.5pt);
\draw [fill=black] (3.0,7.0) circle (1.5pt);
\draw [fill=black] (3.0,6.0) circle (1.5pt);
\draw [fill=black] (3.0,5.0) circle (1.5pt);
\draw [fill=black] (3.0,4.0) circle (1.5pt);
\draw [fill=black] (3.0,3.0) circle (1.5pt);
\draw [fill=black] (3.0,2.0) circle (1.5pt);
\draw [fill=black] (3.0,1.0) circle (1.5pt);
\draw [fill=black] (3.0,0.0) circle (1.5pt);
\end{scriptsize}
\end{tikzpicture}
}
& \subfloat[] {\definecolor{uququq}{rgb}{0.6,0.6,0.6}
\begin{tikzpicture}[line cap=round,line join=round,>=triangle 45,x=1.0cm,y=1.0cm,scale=.65]
\clip(-0,-0.4) rectangle (3.2,8.0);

\draw[color=uququq,fill=uququq,fill opacity=0.1] (0.75,4.25) -- (0.75,2.75) -- (1.25,2.75) -- (1.25,4.25) -- cycle;

\draw (-0,4) node[anchor=north west] {$b_2$};

\draw [dotted] (1.0,3.0)-- (3.0,3.0);
\draw [dotted] (1.0,2.0)-- (3.0,2.0);
\draw [dotted] (1.0,1.0)-- (3.0,1.0);

\draw [] (1.0,6.95)-- (3.0,6.95);
\draw [] (1.0,5.0)-- (3.0,6.0);
\draw [] (1.0,4.0)-- (3.0,5.0);

\begin{scriptsize}
\draw [fill=black] (1.0,1.0) circle (1.5pt);
\draw [fill=black] (1.0,2.0) circle (1.5pt);
\draw [fill=black] (1.0,3.0) circle (1.5pt);
\draw [fill=black] (1.0,4.0) circle (1.5pt);
\draw [fill=black] (1.0,5.0) circle (1.5pt);
\draw [fill=black] (1.0,6.0) circle (1.5pt);
\draw [fill=black] (1.0,7.0) circle (1.5pt);
\draw [fill=black] (3.0,7.0) circle (1.5pt);
\draw [fill=black] (3.0,6.0) circle (1.5pt);
\draw [fill=black] (3.0,5.0) circle (1.5pt);
\draw [fill=black] (3.0,4.0) circle (1.5pt);
\draw [fill=black] (3.0,3.0) circle (1.5pt);
\draw [fill=black] (3.0,2.0) circle (1.5pt);
\draw [fill=black] (3.0,1.0) circle (1.5pt);
\draw [fill=black] (3.0,0.0) circle (1.5pt);
\end{scriptsize}
\end{tikzpicture}
}
& \subfloat[] {\definecolor{uququq}{rgb}{0.6,0.6,0.6}
\begin{tikzpicture}[line cap=round,line join=round,>=triangle 45,x=1.0cm,y=1.0cm,scale=.65]
\clip(-0,-0.4) rectangle (3.2,8.0);

\draw[color=uququq,fill=uququq,fill opacity=0.1] (0.75,2.25) -- (0.75,0.75) -- (1.25,0.75) -- (1.25,2.25) -- cycle;

\draw (-0,2) node[anchor=north west] {$b_3$};


\draw [] (1.0,6.95)-- (3.0,6.95);
\draw [] (1.0,5.0)-- (3.0,6.0);
\draw [] (1.0,4.0)-- (3.0,5.0);
\draw [] (1.0,3.0)-- (3.0,2.0);
\draw [] (1.0,2.0)-- (3.0,1.0);
\draw [] (1.0,1.0)-- (3.0,0.0);

\begin{scriptsize}
\draw [fill=black] (1.0,1.0) circle (1.5pt);
\draw [fill=black] (1.0,2.0) circle (1.5pt);
\draw [fill=black] (1.0,3.0) circle (1.5pt);
\draw [fill=black] (1.0,4.0) circle (1.5pt);
\draw [fill=black] (1.0,5.0) circle (1.5pt);
\draw [fill=black] (1.0,6.0) circle (1.5pt);
\draw [fill=black] (1.0,7.0) circle (1.5pt);
\draw [fill=black] (3.0,7.0) circle (1.5pt);
\draw [fill=black] (3.0,6.0) circle (1.5pt);
\draw [fill=black] (3.0,5.0) circle (1.5pt);
\draw [fill=black] (3.0,4.0) circle (1.5pt);
\draw [fill=black] (3.0,3.0) circle (1.5pt);
\draw [fill=black] (3.0,2.0) circle (1.5pt);
\draw [fill=black] (3.0,1.0) circle (1.5pt);
\draw [fill=black] (3.0,0.0) circle (1.5pt);
\end{scriptsize}
\end{tikzpicture}
}
\end{tabular}
\captionsetup{width=.83\textwidth}
\caption{ An example with 3 seminars, $b_1, b_2, b_3$.  
(a) Two assignments $M_S$ (dashed edges) and $M_T$ (dotted edges); the three alternating paths formed by $M_S \cup M_T$ (light gray).  $q(P_1) = b_1$ because it only intersects vertices from $V_{b_1}$; $q(P_2) = b_1$ because $P_2$ contains a vertex $V(M_T) \setminus V(M_S)$ that is in $V_{b_1}$;  $r(P_3) = b_2$.  (b), (c) and (d) assignments for seminar selections $S \oplus (b_1, 3)$, $S \oplus (b_2, 2)$ and $S \oplus (b_3, 2)$ combining edges of $M_S$ and $M_T$. 
}
\label{fig:path_partition}
\end{figure}
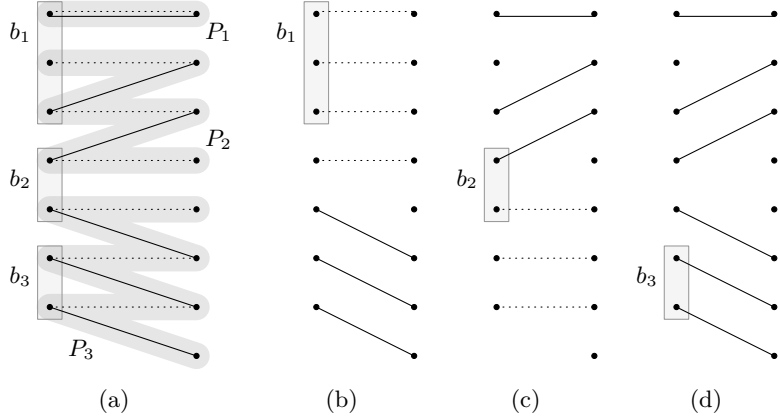

Let $q:\mathcal{C} \rightarrow B$ be a map of the isolated components to the seminars with the following properties: 
\begin{enumerate}
\item $q(P) \in \{b \in B: V(P) \cap V_b \neq \emptyset \}$;
\item if $P \in \calP_b $ for any $b \in B$, $q(P) = b$.
\end{enumerate}
Since $\calP_b$ are disjoint by the previous claim and since for any seminar $b$ it holds by definition that $V(P) \cap V_b \neq \emptyset$ whenever $P \in \calP_b$, it is clear that such a mapping $q$ exists.

For every $b \in B$ let $M_b$ be the matching of $G$ that uses all the edges of $M_T$ from the alternating paths $P \in \mathcal{C}$ mapped by $q$ to the seminar $b$, and all the edges of $M_S$ from the paths $P \in \mathcal{C}$ mapped by $q$ to some other seminar:
$$M_b=[ M_T \cap E(q^{-1}(b)) ]\cup [M_S \cap (E(\mathcal{C}) \setminus E(q^{-1}(b)) )].$$

Observe that any edge of $M_T$ belongs to at least one matching $M_b$ for some $b \in B$ and that any edge of $M_S$ belongs to all but one of the matchings $M_b$.  Therefore, 
$$\sum_{b \in B} \wght(M_b) \geq \wght(M_T) + (|B| - 1) \cdot \wght(M_S).$$

Moreover, observe that for each $b\in B$, $M_b$ is a matching in the bipartite representation of the seminar selection $S \oplus (b, T(b))$. Therefore $p(S \oplus (b, T(b))) = \wght(M_b)$ and the lemma follows.
\end{proof}

\begin{lemma}\label{lem:incremental}
Let $S$ and $T$ be two seminar selections such that $S \oplus (b, T(b))$ is feasible for every $b \in B$.  Let $S^* = \arg\max_{S' \in inc(S)} (p(S') - p(S)) / (c(S') - c(S))$.
Then it holds that: 
$$
\frac{p(S^*) - p(S)}{c(S^*) - c(S)} \geq \frac{p(T) - p(S)}{c(T)}.
$$
\end{lemma}
\begin{proof}
By Lemma \ref{lem:difference} we have that 
\begin{align}\label{eq:bound_opt}
\sum_{b \in B} [p(S \oplus (b, T(b))) - p(S)] \geq p(T) - p(S).
\end{align}

Since
$\sum_{b \in B} [c(S \oplus (b, T(b))) - c(S)] \leq \sum_{b\in B} T(b) = c(T)$, inequality (\ref{eq:bound_opt}) implies that 
\begin{align}\label{eq:opt_selection}
\frac{\sum_{b \in B} [p(S \oplus (b, T(b))) - p(S)]}{ \sum_{b \in B} [c(S \oplus (b, T(b))) - c(S)]} 
\geq \frac{p(T) - p(S)}{c(T)}. 
\end{align}

Then, there exists at least one seminar $b^* \in B$ such that
\begin{align}\label{eq:single_opt_selection}
\frac{p(S \oplus (b^*, T(b^*))) - p(S)}{ c(S \oplus (b^*, T(b^*))) - c(S)} 
\geq \frac{p(T) - p(S)}{c(T)}. 
\end{align}

Since $S \oplus (b^*, T(b^*)))$ is clearly in $inc(S)$ the lemma follows directly from Eq. (\ref{eq:single_opt_selection}) and the definition of $S^*$.
\end{proof}

\begin{lemma}\label{lem:bound_si}
Let $T$ be a feasible seminar selection and let $r \in \mathbb{N}$ be such that $S_i \oplus (b, T(b))$ is feasible for every $i < r$ and $b \in B$.  Then for each $i \leq r$ the following holds: 
$$
p(S_i) - p(S_0) \geq \left[1 - \prod_{k=0}^{i-1} \left(1 - \frac{c(S_{k+1}) - c(S_k)}{c(T)} \right) \right] \cdot \Big(p(T)-p(S_0)\Big).
$$
\end{lemma}
\begin{proof}
We prove the lemma by induction on the iterations $i$.  By the definition of the algorithm, $S_1$ is the seminar selection with maximum marginal density in $inc(S_0)$, and thus Lemma \ref{lem:incremental} shows that the inequality holds for $i = 1$.  Suppose that the lemma holds for iterations $1, ..., i$.  We show that it also holds for iteration $i + 1$.  For ease of exposition, for the remainder of this proof let $\alpha_i = \frac{c(S_{i+1}) - c(S_i)}{c(T)}$.  

\begin{align*}
p(S_{i+1}) - p(S_0)&= p(S_i) - p(S_0) + p(S_{i+1}) - p(S_i) \\
		&\geq p(S_i) - p(S_0) + \alpha_i \cdot (p(T) - p(S_i)) \\
		&= (1 -  \alpha_i )p(S_i)   + \alpha_i \cdot p(T) - p(S_0)\\
		&\geq (1 -  \alpha_i )\cdot\left(1 - \prod_{k=0}^{i-1} \left(1 - \alpha_k \right) \right)(p(T) - p(S_0))\\
		&\quad \quad + (1 -  \alpha_i ) \cdot p(S_0) + \alpha_i \cdot p(T) - p(S_0)\\
		&= \left(1 - \alpha_i -  \prod_{k=0}^{i} \left(1 - \alpha_k \right) \right)(p(T) - p(S_0))\\
		&\quad \quad + \alpha_i \cdot (p(T) - p(S_0))\\
		&=\left(1 -  \prod_{k=0}^{i} \left(1 - \alpha_k \right) \right)(p(T) - p(S_0)).
\end{align*}	
Where the first inequality follows from Lemma \ref{lem:incremental} and the second inequality follows from the induction hypothesis.
\end{proof}

\begin{theorem}
When $S_0$ is the empty assignment the {\bf \textsc{Greedy}} algorithm is a $\frac{1}{2}\cdot \left(1 - e^{-1} \right)$ approximation for SAP.
\end{theorem}
\begin{proof}
Let $OPT$ be the seminar selection of a fixed optimal assignment solution for the given SAP instance.  Let $b^*\in B$ be the seminar that is allocated the most students in $OPT$ and let $OPT'$ be the seminar selection for which $OPT'(b^*) = 0$ and $OPT'(b) = OPT(b)$ for any $b\neq b^* \in B$.  Let $r$ be the first iteration of the algorithm for which $c(S_r) > c(OPT')$.  Clearly, $S_i \oplus (b, OPT(b))$ is feasible for every $i < r$ and $b \in B$.  Since $p(S_0) = 0$, by applying Lemma \ref{lem:bound_si} to iteration $r$ we obtain: 
\begin{align}\label{eq:bound2}
p(S_r) 	&\geq \left[1 - \prod_{k=0}^{r-1} \left(1 - \frac{c(S_{k+1}) - c(S_k)}{c(OPT')} \right) \right] \cdot p(OPT') \nonumber\\
			&\geq \left[1 - \prod_{k=0}^{r-1} \left(1 - \frac{c(S_{k+1}) - c(S_k)}{c(S_r)} \right) \right] \cdot p(OPT').
\end{align}

Observe that $c(S_r) = \sum_{k=0}^{r-1} c(S_{k+1}) - c(S_k)$ and that for any real numbers $a_0, ..., a_{r-1}$ with $\sum_{k=0}^{r-1} a_k = A$ it holds that:
\begin{align}\label{eq:e_bound}
\prod_{k=0}^{r-1}\left(1 - \frac{a_k}{A}\right) \leq \left(1 - \frac{1}{r}\right)^r < e ^{-1}.
\end{align}

Therefore Eq. (\ref{eq:bound2}) implies $p(S_r) > (1 - e^{-1})\cdot p(OPT')$.  Since the profit of $\mathcal{A}_2$ is at least $p(b^*, OPT(b^*))$ it holds that 
\begin{align*}
\mathcal{A}_1 + \mathcal{A}_2 &> (1 - e^{-1})\cdot p(OPT') + p(b^*, OPT(b^*)) \\
		&\geq (1 - e^{-1})\cdot p(OPT)
\end{align*}
and therefore either $\mathcal{A}_1$ or $\mathcal{A}_2$ has profit at least $\frac{1}{2}\cdot(1 - e^{-1}) p(OPT)$.
\end{proof}

\section{Improving the Approximation}\label{sec:greedy_plus}
In this section we show that the algorithm can be improved by starting the greedy algorithm not from an empty seminar selection, but from a seminar selection that is part of the optimal solution.  The improved algorithm is less efficient but achieves the optimal approximation ratio of $(1 - e^{-1})$.  Let $\calA_{opt}$ be an optimal seminar assignment and for any $b \in B$ let $p_{opt}(b)$ be the profit obtained in this assignment from seminar $b$: 
$$
p_{opt}(b) = \sum_{i \in \calA_{opt}^{-1}(b)} p(i, b).
$$
Clearly, the profit of the optimal solution is $\sum_{b\in B} p_{opt}(b)$.   W.l.o.g, let $b_1, b_2, b_3$ be the three seminars of the optimal solution with highest profit and let $S^*$ be a seminar selection such that $S^*(b) = OPT(b)$ if $b \in \{b_1, b_2, b_3\}$, and $S^*(b) = 0$ otherwise. 

\begin{theorem}
When $S_0 = S^*$ the {\bf \textsc{Greedy}} algorithm is a $\left(1 - e^{-1} \right)$-approximation for SAP.
\end{theorem}
\begin{proof}
Let $OPT$ be the seminar selection corresponding to $\calA_{opt}$.  Let $b^*$ be the seminar that is allocated the most students in $OPT$ and is not allocated students in $S^*$.  Moreover, let $OPT'$ be the seminar selection for which $OPT'(b^*) = 0$ and $OPT'(b) = OPT(b)$ for any $b\neq b^* \in B$.  Let $r$ be the first iteration of the algorithm for which $c(S_r) > c(OPT')$.  Clearly, the seminar selection $S_i \oplus (b, OPT(b))$ is feasible for every $i < r$ and $b \in B$.  By applying Lemma \ref{lem:bound_si} to iteration $r$ we obtain: 
\begin{align*}
p(S_r) - p(S^*) 	&\geq \left[1 - \prod_{k=0}^{r-1} \left(1 - \frac{c(S_{k+1}) - c(S_k)}{c(OPT')} \right) \right] \cdot \Big(p(OPT') - p(S^*)\Big) \nonumber\\
			&\geq \left[1 - \prod_{k=0}^{r-1} \left(1 - \frac{c(S_{k+1}) - c(S_k)}{c(S_r)} \right) \right] \cdot \Big(p(OPT') - p(S^*)\Big).
\end{align*}

By applying Eq. (\ref{eq:e_bound}) we obtain that 
\begin{align*}
p(S_r) - p(S^*)\geq (1 -1/e)\cdot \Big(p(OPT') - p(S^*)\Big) , 
\end{align*}
and therefore
\begin{align}
p(S_r) &\geq (1 - 1/e)\cdot p(OPT') + p(S^*) / e \nonumber \\
	  & \geq (1 - 1/e)\cdot p(OPT) - p_{opt}(b^*) + p(S^*) / e.
\end{align}

By hypothesis $S^*$ selects the three seminars with maximum profit in the optimal assignment and allocates exactly as many students to each as $OPT$ does.  Then, since $p_{opt}(b^*) \leq p_{opt}(b_i)$ for $i = 1, ..., 3$ it holds that $p(S^*) \geq 3\cdot p_{opt}(b^*) > e \cdot p_{opt}(b^*)$ and the theorem follows.
\end{proof}

Observe that the number of feasible seminar selections assigning students to at most three seminars is polynomial in the size of the input.  Therefore, by repeatedly calling the greedy algorithm with all possible such selections our main result follows:

\begin{corollary}
There exists a polynomial time $(1 - e^{-1})$-approximation algorithm for SAP.
\end{corollary}

\newpage
\bibliographystyle{plain}	
\bibliography{bibliography}

\end{document}